\documentclass[11pt,a4paper]{article}
\usepackage[latin1]{inputenc}
\usepackage[english]{babel}
\usepackage{amsfonts}
\usepackage{amssymb}
\usepackage{amscd}
\usepackage{amsmath}
\usepackage{theorem}
\usepackage{graphicx}
\usepackage{epstopdf}
\usepackage{marvosym}
\usepackage{color}
\usepackage{geometry}
\usepackage[hang,flushmargin]{footmisc}

\newcommand{\R}{\mathbb{R}}

\newcommand{\mP}{\mathbb{P}}
\newcommand{\mE}{\mathbb{E}}
\newcommand{\md}{\,{\rm d}}
\newcommand{\s}{\sum\limits}

\newcommand{\one}{1\mkern-5mu{\hbox{\rm I}}}

\newcommand{\E}{\mathbb{E}}
\newcommand{\Prob}{\mathbb{P}}
\newcommand{\td}{\mathrm{d}}

 \newcommand\hp{H_\alpha^\prime}
 
\newcommand\gqw{G_{w}}

\def\WT{\widetilde}

\theoremstyle{break}
\theorembodyfont{}
\newtheorem{Def}{Definition}[section]
\newtheorem{remark}[Def]{Remark}
\newtheorem{lemma}[Def]{Lemma}
\newtheorem{proposition}[Def]{Proposition}
\newtheorem{corollary}[Def]{Corollary}
\newtheorem{theorem}[Def]{Theorem}
\newtheorem{example}[Def]{Example}
\makeatletter
\newenvironment{proof}{\noindent{\textit{Proof:}}}{%
\unskip\nobreak\hfil\penalty50\hskip1em\null\nobreak
$\Box$
\parfillskip=\z@\finalhyphendemerits=0\endgraf\bigskip}

\let\oldendexample\endexample
\def\endexample{\unskip\nobreak\hfil\penalty50\hskip1em\null\nobreak\hfil
$\blacksquare$\parfillskip=\z@\finalhyphendemerits=0\endgraf\oldendexample}
\begin{document}

\title{Optimal Dividends Paid in a Foreign Currency for a L\'evy Insurance Risk Model}

\author{Julia Eisenberg \quad\quad \quad\quad   Zbigniew Palmowski\footnote{\Letter\; zbigniew.palmowski@gmail.com} \vspace{0.2cm}\\\hspace{0.5cm}\footnotesize{TU Wien \quad\quad\quad\quad\hspace{1.7cm}  Wroc\l aw University}\vspace{-0.1cm} \\\hspace{4.8cm}\footnotesize{of Science and Technology}}

\date{}

\maketitle

\begin{abstract}
\noindent
This paper considers an optimal dividend distribution problem for an
insurance company where the dividends are paid in a foreign currency.
In the absence of dividend payments, our risk process follows a spectrally negative L\'evy process.
We assume that the exchange rate is described by a an exponentially L\'evy process, possibly containing the same risk sources like the surplus of the insurance company under consideration.
The control mechanism chooses the amount of dividend payments. The objective is to maximise the expected dividend payments received until the
time of ruin and a penalty payment at the time of ruin, which
is an increasing function of the size of the shortfall at ruin.
A complete solution is presented to the corresponding stochastic
control problem.
Via the corresponding Hamilton--Jacobi--Bellman equation we find the necessary and sufficient conditions for optimality of a single dividend barrier strategy.
A number of numerical examples illustrate the theoretical analysis.
\vspace{6pt}
\settowidth\labelwidth{{\it 2010 Mathematical Subject Classification: }}%
                \par\noindent {\it 2010 Mathematical Subject Classification: }%
                \rlap{Primary}\phantom{Secondary}
                60G51, 93E20\newline\null\hskip\labelwidth
                Secondary 91B30
\end{abstract}

\section{Introduction}
In the public eye, dividend payments are holding the title to be one of the most important signs of financial health and future stability of shares-issuing companies. Thus, a forecast of opulent future dividends, compared to a benchmark such as ten-year Government bonds, will most likely attract new investors, clients and business partners. Therefore, it is natural to consider future dividend payments as a risk measure quantifying company's future profitability and debt sustainability. Since the pathbreaking work of Bruno de Finetti in 1957, \cite{finetti}, substantial research has been carried out on finding the optimal dividend strategy in the framework of the classical risk model or diffusion approximation as a surplus process for an insurance company. The survey \cite{alth} sums up the most important results for these types of surplus. Avram et al.\ \cite{APP} generalised de Finetti's problem to spectrally negative L\'evy processes as surplus. Loeffen extended their results in \cite{Loeffen1} and added transaction costs in \cite{Loeffen2}. Loeffen and Renaud \cite{LR} modified the optimisation problem by adding an affine penalty function at ruin.
\\Despite severe differences in modelling the surplus and additional constraints, the above works have one feature in common: the discounting factor or rather the preference rate. The preference rate is usually assumed to remain constant and positive over time, signalising the setup ``money today is more preferable to money tomorrow''. However, in the times of negative interest rates, like nowadays, a perpetual positive interest rate will lead to deterioration of results. For instance, Akyildirim et al.\ \cite{ak}, Eisenberg \cite{eis}, Jiang and Pistorius \cite{jiang} incorporated stochastic interest rate into the dividend optimisation framework.
\medskip
\\
Another aspect that has not been studied until now in the framework of dividend maximisation are the foreign interest rates. Big insurance companies have clients and shareholders all over the world.
For instance, top global reinsurance companies, including such giants like Munich Re and Swiss Re, have established themselves in the Middle East more than a decade ago and have been recently expanding to Asia and Latin America -- while ``local'' reinsurance companies are still ``in the cradle''.
In the most cases, the dividends are declared in the domestic currency of companies or in US dollars and are paid to the shareholders in the local currency using the actual exchange rate.\medskip
\\
Currency fluctuations are a natural consequence of the floating exchange rate system (i.e.\ a currency's value is allowed to fluctuate in response to foreign-exchange market events), which is used in the most economies. Indeed just a few countries worldwide are currently using the fixed rate approach, where the domestic currency is pegged to a stronger currency or a basket of them. Many factors impact a foreign exchange rate, for instance relative supply and demand of the two currencies, a forecast for inflation etc. Thus, any noticeable changes in the underlying economy affect the exchange rates and the economic activities of almost all domestic market participants. For shareholders of an insurance company such events might become crucial as the affected company can decide to shorten dividend payments due to an unfavourable market situation. Thus, the metamorphoses with the exchange rate and the shortening of dividends can have the same risk component. Note that the impact type described above is rather of a continuous nature, reflecting infinitesimally economic changes on the daily basis.
\\On the other hand, in recent years, a number of incidents known as ``flash crashes'' has been shivering the global financial market. Christensen et al.\ state in \cite{reno} that the number of flashes will be even increasing. The sudden market crashes will again affect both -- the exchange rates and the insurance companies. These changes are of a jump nature, occurring at discrete times but on the regular basis.
\medskip
\\In the present paper we describe the surplus process of an insurance company by a L\'evy process containing a diffusion and a jump part. We assume that the insurance company under consideration targets to maximise the expected discounted amount of dividends paid in a foreign currency. The exchange rate is assumed to follow a L\'evy process featuring a dependence on the surplus process. Since L\'evy processes can be decomposed into a diffusion part and a jump part,   we distinguish two cases for dependencies: dependence of the continuous and jump parts.
The paper is organized as follows.
In Section \ref{sec:model}, we introduce the basic
  notation and describe the model we deal with. Section \ref{sec:pre} is dedicated to the
  related one-sided and two-sided problems.
In Section \ref{sec:main}, we present the
  Verification Theorem, necessary and sufficient conditions for the barrier strategy to be
  optimal.
For the sake of clarity of presentation, we shifted the proofs to Section \ref{sec:proof}. Section \ref{sec:examples} presents two detailed examples.
\section{The Model}\label{sec:model}
Recall the classical Cram\'er-Lundberg model
\begin{equation}\label{classicalrisk} R_t -R_0=
 ct- S_t, \quad S_t=\sum_{k=1}^{N_t} C_k,
\end{equation}
which is used in collective risk theory to describe the surplus $R=\{R_t, t\in\mathbb{R}_+\}$
of an insurance company. Here, $C_k$ are i.i.d. positive random
variables representing the claims and $S_t$ denotes the aggregate claims up to time $t$. The claim number $N=\{N_t, t\in\mathbb{R}_+\}$ is modelled via a homogeneous Poisson process with intensity $\lambda$ and is independent of the claims. Finally, $c$ represents the premium rate fulfilling $c>\lambda m>0$ and $m=\mE[C_1]<\infty$, in order to allow the process to remain non-negative with a positive probability.
\\
In 1973 Gerber \cite{gerber} introduced some uncertainty into the Cram\'er-Lundberg model by adding a Brownian motion. The ``perturbed model" is then $R_t-R_0:= \sigma B_t +c t - S_t $
where $B_t$ denotes a standard Brownian motion, describing small random
fluctuations of the surplus.
\\
A further very important generalization is to replace the aggregate claim amount $S$ by a general subordinator (a non-decreasing L\'{e}vy process, with L\'evy measure $\nu_R(\td
x), x \in \R_+,$ which may have infinite mass).  Under this model, the ``fluctuations'' can arise either continuously, due to the Brownian motion, or due to the infinite jump-activity.
\\
Assuming $S$ to be a pure jump-martingale with i.i.d.\ increments and
negative jumps with L\'{e}vy measure $\nu_R(\td x)$, one arrives
thus to a general integrable spectrally negative L\'{e}vy process
$R=\{R_t,~t\in\mathbb{R}_+\}$ i.e.\ a stochastic process with stationary independent increments, no positive jumps and c\`{a}dl\`{a}g paths with $R_t$ integrable for any $t\geq 0$ and $\mE[R_1]>0$\label{R0} in order for the surplus to be profitable, confer Kyprianou \cite{Kbook} for details. The corresponding L\'evy-Khintchin triple is $(c,\sigma, \nu_R)$ and $R_0=x$, i.e. the generator of $R$ is given by
\[
\mathfrak{A}_1 f(x)=c f'(x)+\frac{\sigma^2}2 f''(x)+\int_{\R}f(x+h)-f(x)-f'(x)h\one_{[|h|\le 1]}\;\nu_R(\td h)\;
\]
\label{generator}
for a suitable function $f$ from the domain of the generator.
\bigskip
\\We further assume that exchange rate process denoted by $Y=\{Y_t,~t\in\mathbb{R}_+\}$
is a L\'evy process with a corresponding triple $(p, \delta, \nu_Y)$ and $Y_0=l$, i.e. the generator of $Y$ has the form
\[
\mathfrak{A}_2 f(l)=p f'(l)+\frac{\delta^2}2 f''(l)+\int_{\R}f(l+h)-f(l)-f'(l)h\one_{[|h|\le 1]}\;\nu_Y(\td h)\;
\]
for a suitable function $f$ from the domain of the generator $\mathfrak A_2$. Note that 
we do not assume that the process $Y$ is a spectrally negative one. Indeed, the discounting factor can evolve into one as well as into the other direction.
\bigskip
\\The both processes are defined on some common probability space $(\Omega,\mathcal{F},\{\mathcal{F}_t\}_{\{t\geq 0\}}, \Prob)$,
where ${\mathcal{F}}=\{\mathcal{F}_t\}_{\{t\geq 0\}}$ is the natural
filtration satisfying the usual conditions of right-continuity and
completeness generated by bivariate L\'evy process
$X=\{X_t:=(R_t,Y_t),~t\in\mathbb{R}_+\}$. {To avoid degeneracies, we exclude the case that $R$ or $Y$ have monotone paths.} We denote by $\nu(\td z,\td y)$ the jump measure of the process $X$. Note that
\[
\nu(\td z, (-\infty,\infty))=\nu_R(\td z)\quad\text{and}\quad \nu([0,\infty),\td y)=\nu_Y(\td y).
\]
We denote by $\{\Prob_{\underline{x}}=\Prob_{(x,l)}, \underline{x}=(x,l)\in\mathbb{R}^2\}$ the family of probability measures that correspond to
the translations of $X$ by a vector, that is, $\Prob[X_0=\underline{x}]=1$.
Later, when it will be clear, we skip underlying of $x$ to note the only dependence on $x$.
In this case by $\E_{\underline{x}}$ and $\E_x$ we denote the corresponding expectations.
Finally, we will use the notation $\Prob_0=\Prob$ and $\E_0=\E$ as well.
\\
To ensure that $R_t$ and $Y_t$ have finite means for fixed $t\geq 0$
the L\'{e}vy measure $\nu$ is assumed to satisfy the integrability condition
\[
\int_{[\mathbb{R}\backslash (-1,1)]^2} ||\underline x|| \;\nu(\td \underline x) < \infty.
\]
As stated in the introduction, the processes $R$ and $Y$ are assumed to be dependent.
Since the continuous part and the jump part of a L\'evy process are independent it is enough to consider the dependence structure of the continuous and the discontinuous part separately.
The generator of the process $X$ in case of both types of dependency is given by
\begin{align*}
&\mathfrak{A} f(x,l)=c f_x(x,l)+\frac{\sigma^2}2f_{xx}(x,l)+p f_l(x,l)+\frac{\delta^2}2f_{ll}(x,l)+\rho\sigma\delta f_{xl}(x,l)
\\&{}+\int_{\R^2}f(x+h_2,l+h_1)-f(x,l)-f_x(x,l)h_1\one_{[|h_1|\le 1]}-f_l(x,l)h_2\one_{[|h_2|\le 1]}\;\nu(\td h_1,\td h_2)\;.\nonumber
\end{align*}
If $R$ and $Y$ depend just over the jump part, $\rho\sigma\delta f_{xl}(l,x)$ disappears. By dependency just over the continuous part, the integral above transforms to
\begin{align*}
\int_{\R}f(x,l+h)-f(x,l)-f_l(x,l)h&\one_{[|h|\le 1]}\;\nu_Y(\td h)
\\&{}+\int_{\R}f(x+h,l)-f(x,l)-f_x(x,l)h\one_{[|h|\le 1]}\;\nu_R(\td h)\;.
\end{align*}
For more details confer for instance \cite{muller}.
\bigskip
\\We assume that the considered insurance company pays dividends and the ex-dividend process is given by
\begin{equation*}
R^\pi_t=R_t-L^\pi_t,
\end{equation*}
where $\pi$ denotes a strategy chosen from the set $\Pi$ of all admissible dividend  controls, resulting in dividend process $L^{\pi}_t$ - denoting the accumulated dividends under $\pi$ paid up to time $t$. An admissible dividend strategy $\pi$ generates the dividend process $L^\pi=\{L^\pi_t, t\in\R_+\}$, which is cadlag, adapted  to the filtration $ \mathcal F=\{\mathcal{F}_t\}_{t\ge0}$, and at any time preceding the ruin, the dividend payment is smaller than the size of the available reserves ($L^\pi_t-L^\pi_{t-}<R_{t-}^\pi$), i.e. the ruin cannot be caused by a dividend payment. \label{admissible}
\medskip
\\
The object of interest is the expected discounted amount of dividends paid in a domestic and declared in a foreign currency
\[
\mathcal{D}(\pi):=\int_0^{T^\pi} e^{- Y_t}\md L^{\pi}_t
\]
and the expected discounted penalty payment (so-called Gerber-Shiu function)
\[
\mathcal{W}(\pi):=e^{- Y_{T^\pi}}w(R_{T^\pi}^\pi)\;.
\]
Here, $T^\pi:=\inf\{t \geq 0: R^{\pi}_t<0\}$ is the ruin time and $w$ is a penalty function acting on the negative half-line.
Later, unless it is necessary we will write $T$ instead of $T^\pi$ to simplify the notation.
\\The process $Y_t$ apart from the interpretation as an exchange rate also describes discounting. In particular if $Y_t=qt$ the $q$ could be interpreted as a given discount or rather a preference rate, describing the monetary preferences of the considered insurance company. Our objective is to maximise
\[
V_\pi(x,l):=\mE_{(x,l)}\big[\mathcal{D}(\pi)\big]+\mE_{(x,l)}\big[\mathcal{W}(\pi)\big]
\]
over all admissible strategies, that is to find the so-called value function
\begin{align}
\label{cost2}
\begin{split}
&V(x,l):=\sup_{\pi\in \Pi}V_\pi(l,x),
\end{split}
\end{align}
and the optimal strategy $\pi_*\in\Pi$, if it exists, such that
\begin{equation*}
V(x,l)=V_{\pi_*}(x,l)\qquad\text{for all }x\geq 0,\, l\in\R.
\end{equation*}
\subsection{Preliminaries}\label{sec:pre}
In this section, we summarise the basic definitions and properties of L\'evy processes and some other concepts we will use in our modelling.
\\We conjecture that the optimal dividend payment strategy will be of a barrier type. It means that the dividends are paid as the excess of the surplus above a certain constant level, say $a>0$. If the surplus is above the level $a$, the excess will be immediately distributed as a lump sum dividend payment and the surplus amounts to $a$. By starting below $a$, the insurance company will not pay any dividends until the surplus attains $a$, the considerations stop if the surplus attains $0$ before attaining $a$. Therefore, we will need the following first passage times
\[
\tau_a^+:=\inf\{t\geq 0: R_t\geq a\}\quad\text{and}\quad \tau_0^-:=\inf\{t\geq 0: R_t<0\}\;.
\]
We will define now formally auxiliary functions $\Delta$ for which the following exit identity holds true
\begin{gather}
\label{exit1}
\E_{(x,l)}\left[e^{ -Y_{\tau_a^+}}\one_{[ \tau_a^+<\tau^-_0]}\right]= \frac{\Delta(x)}{\Delta(a)}e^{-l},
\end{gather}
where $x\in(0,a)$.
\\
In the following, we recall some results from the fluctuation theory for spectrally negative L\'{e}vy processes. For more details, confer \cite{Kbook,Sato} and references therein.
\bigskip
\\
Let
\begin{equation*}
\E_{(x,l)}[e^{\langle \theta,X_t\rangle}] = e^{t\psi(\theta)+\theta_1 x+\theta_2 l}
\end{equation*}
for $\theta=(\theta_1, \theta_2)\in D \subseteq  \R_+ \times \R$ and some set $D$ for which above expectation is well-defined and $\langle .,.\rangle$ is a scalar product.
For any $\theta\in D$ we denote by $\Prob^\theta$ an exponential tilting of measure $\Prob$ with Radon-Nikodym derivative
$\Prob$ given by
\begin{equation*}
\left.\frac{\td \Prob^\theta}{\td\Prob}\right|_{\mathcal{F}_t} = \exp\Big\{\langle \theta,X_t\rangle - \psi(\theta)t\Big\}.
\end{equation*}
Under the measure $\Prob^\theta$ the process $X$ is still a bivariate L\'{e}vy process
with the Laplace exponent
$\phi_\theta(s)$ with $s\in\R^2$ given by:
\begin{equation}
\label{tildepsi}
\phi_\theta(s)=\log\Big(\mE^{\theta}\Big[e^{\langle s, X_1\rangle}\Big]\Big)=\psi(s+\theta)-\psi(\theta)\;,
\end{equation}
where $\E^\theta$ denotes the expectation with respect to $\Prob^\theta$. From now on we assume that there exists an $\alpha \geq 0$ such that
\begin{equation}
\label{alphadef}
\psi(\alpha, -1)=0\;.
\end{equation}
We denote by
$$\psi_R(\beta):=\phi_{(\alpha, -1)}(\beta, 0)=\log\Big(\E^{(\alpha, -1)} [e^{\beta R_1}]\Big)$$
the Laplace exponent of $R$ under $\Prob^{(\alpha, -1)}$.
Note that under $\Prob^{(\alpha, -1)}$ the process $R$ has the following L\'evy-Khintchin triple:
\begin{align}
\label{newtriple}
(\tilde c, \sigma, \mu_R)\;,
\end{align}
where
\begin{align*}
&\tilde c:=c+\alpha \sigma^2-\rho\sigma\delta+\int_{\R^2} \Big\{e^{\alpha h_1-h_2}-1\Big\}h_1\one_{[|h_1|\le 1]}\;\nu(\td h_1,\td h_2)\;,
\\&\mu_R(A): =\int_{A\times\R} e^{\alpha h_1-h_2}\;\nu(\td h_1,\td h_2)\;\mbox{for all Borel sets $A$}\;.
\end{align*}
Further, there exists a function $W^{\alpha}: [0,\infty) \to [0,\infty)$, called the scale function, confer for instance \cite{bert}, continuous and increasing with Laplace transform
\begin{equation}
\label{eq:defW}
\int_0^\infty
e^{-\beta y} W^{\alpha}(y)  \md y = \psi_R(\beta)^{-1}\;.
\end{equation}
The domain of $W^{\alpha}$ is extended to the entire real axis by setting $W^{\alpha}(z)=0$ for $z<0$.
\subsubsection{Assumption 1:} Throughout the paper we assume that
the following (regularity) condition is satisfied:
\begin{equation}\label{condW1}
W^\alpha\in \mathcal C^2(0,\infty).
\end{equation}
To get it we can assume that either
\begin{equation}\label{eq:condW}
\mu_R(-\infty,-x),\;\; x\ge0\;\; \mbox{has a completely monotone density;
}
\end{equation}
see \cite[p. 695]{Chan} and \cite{Loeffen1}\footnote{A function $f$ with the domain $(0,\infty)$ is said to be completely monotone, if the derivatives $f^{(n)}(x)$ exist for all $n= 0,1,2,3,...$ and $(-1)^nf^{(n)}(x)\ge 0$ for all $x>0$.\label{foot}}, or
\[\sigma^2>0;\]
see \cite[Thm. 1]{Chan}, or that $R_t$ is given in \eqref{classicalrisk} with
\begin{align*}
\mu_R(-\infty,-x)\in \mathcal C^1(0,\infty);
\end{align*}
see \cite[Problem 8.4 (ii)]{Kbook}.
\bigskip
\\The function $W^\alpha$ plays a key role in the solution of the two-sided exit problem as shown by the following classical identity:
\begin{equation}
\label{twoex}
\Prob^{(\alpha, -1)}_{x}\left[\tau^-_0 > \tau^+_a\right] = \frac{W^\alpha(x)}{W^\alpha(a)} \end{equation}
that holds for $x\in[0,a]$, see \cite{Kbook}. The function $\Delta$ defined in \eqref{exit1} is related to the above scale function $W^{\alpha}$ in following way.
\begin{lemma}\label{Wrepr}
It holds
\begin{equation*}\label{rel1}
\Delta(z)= e^{\alpha z}W^\alpha(z).
\end{equation*}
\end{lemma}
\begin{proof}
Note that by \eqref{twoex} we have
\begin{equation*}
\E_{(x,l)}\left[e^{-Y_{\tau_a^+}}\one_{[ \tau_a^+<\tau^-_0]}\right]=
e^{\alpha(x-a)}e^{-l}\cdot \Prob^{(\alpha, -1)}\left[\tau_a^+<\tau^-_0\right]
=e^{-l}\cdot \frac{e^{\alpha x}W^\alpha(x)}{e^{\alpha a}W^\alpha(a)}
\end{equation*}
which completes the proof.
\end{proof}
In order for the optimisation problem to be well-defined, we require the following condition.
\subsubsection{Assumption 2:}
\begin{equation}
\psi(0,-1)<0\;.\label{assump:2}
\end{equation}
We first show that without this assumption the value function could be infinite. Indeed, let $\psi(0,-1)>0$ and we assume w.o.l.g. $w=0$ and $\mE[R_1]<\infty$. Letting $b:=\mE[R_1]/2$ and defining $\pi^b$ to be the strategy with the dividend payout $L_t^{\pi^b}=bt$ yields using Tonelli's theorem
\begin{align*}
V(x,l)&\ge V_{\pi^b}(x,l)=b \E_{(x,l)} \bigg[\int_0^{T^{\pi^b}}e^{-Y_t}\md t\bigg]=b  \int_0^\infty \E_{(x,l)}\Big[e^{-Y_t}\one_{[T^{\pi^b}>t]}\Big]\md t
\\&=b  e^{-l}\int_0^\infty e^{\psi(0,-1)t}\cdot \mP_x^{(0,-1)}\Big[T^{\pi^b}>t\Big]\md t\;.
\end{align*}
Because the ex-dividend process fulfils $\mE[R_1^{\pi^b}]>0$, confer the definition of $R$ on p.\ \pageref{R0} it holds $\mP[T^{\pi^b}=\infty]>0$ and accordingly $\mP_x^{(0,-1)}\big[T^{\pi^b}=\infty\big]>0$.
Thus we immediately get $V(x,l)=\infty$.
\\
On the other hand, under Assumption \eqref{assump:2} our value function is well-defined. Indeed,
this assumption yields that $\alpha>0$ for $\alpha$ solving \eqref{alphadef}.
Therefore, due to the continuity of $\psi$ there is an $a\in(0,\alpha)$ such that $\psi(a,-1)<0$. Then,
\begin{align*}
V(x,l)&\le \sup\limits_{\pi\in\Pi}\E_{(x,l)} \bigg[\int_0^{T^{\pi}}e^{-Y_t}R_t\md t\bigg]\le \sup\limits_{\pi\in\Pi}\int_0^\infty \E_{(x,l)} \big[e^{-Y_t}R_t\one_{[T^{\pi}>t]}\big]\md t
\\&\le \frac 1a\sup\limits_{\pi\in\Pi}\int_0^\infty \E_{(x,l)} \big[e^{a R_t-Y_t}\one_{[T^{\pi}>t]}\big]\md t
=\frac{e^{-l}}a \sup\limits_{\pi\in\Pi}\int_0^\infty e^{\psi(a,-1)t}\E_{x}^{(a,-1)} \big[\one_{[T^{\pi}>t]}\big]\md t
\\&\le \frac{e^{-l}}a \sup\limits_{\pi\in\Pi}\int_0^\infty e^{\psi(a,-1)t}\md t<\infty\;.
\end{align*}
\subsubsection{Penalty functions}
Throughout the paper we will also consider the penalty functions belonging to the family of functions
$\mathcal{R}$ which is defined in the following way.
$\mathcal{R}$ is the set of c\`adl\`ag
functions
$w:(-\infty,0]\to\mathbb{R}$ that are left-continuous at~$0$,
admit a finite first
left-derivative $w_-'(0)$ at $0$, and satisfy the
integrability condition
\begin{eqnarray*}
&& \sup_{y>1}  \int_{[y,\infty)}\sup_{u\in[y-1,y]}|w(u-z)|e^{\alpha z}\;\nu_R(\td z) < \infty.
\end{eqnarray*}
Let
\begin{equation}
\label{Gqw}
\gqw (x):=e^l\cdot \E_{(x,l)}\left[e^{-Y_{\tau_0^-}}w(R_{\tau_0^-})\right]= \E_x^{(\alpha, -1)}\left[w(R_{\tau_0^-})\right].
\end{equation}
For $w\in\mathcal{R}$, from Proposition 4.9 in Avram et al. \cite{APP2}, we have
the following lemma.
\begin{lemma}\label{Grepr}
Let $w\in\mathcal{R}$. For any $x\in\mathbb{R}$ it holds
\begin{eqnarray*}
&& \gqw (x) = F_w(x) - W^{\alpha}(x)\kappa_w,\quad\text{with}\\
&& \kappa_w := \left[\frac{\sigma^2}{2}w'(0-) + \frac{1}{\E^{(\alpha, -1)}[ R_1]}w(0) -
\mathcal{L}w_\nu\right],\nonumber
\end{eqnarray*}
where $\mathcal{L}w_\nu=\int_0^\infty \int_x^\infty[w(x-z)-w(0)]e^{\alpha z}\;\nu_R(\md z) \md x$
and the function $F_{w}:\mathbb{R}\to\mathbb{R}$ is given by $F_w(x)=w(x)$ for
$x< 0$, and by
\begin{align*}
& F_w(x) = w(0) + w_-'(0) x - \int_0^x W^{\alpha}(x-y) J_w(y) \md y, \; x\in\mathbb{R}_+,\ \text{with}\\
& J_w(x) = w_-'(0) c+ \int_{x}^\infty\{w(x-z) - w(0) +
w_-'(0)(z-x)\}e^{\alpha z}\;\nu_R(\td z)\;.
\end{align*}
\end{lemma}
\section{Main Results}\label{sec:main}
In this section we will present the main result of the paper, namely show that the optimal strategy among all admissible strategies $\Pi$, defined on p.\ \ref{admissible}, is of a constant barrier type. For that purpose, we consider the corresponding Hamilton--Jacobi--Bellman (HJB) equation, which has been derived using heuristic arguments, see for instance \cite{Schmbook}.
\begin{align}\label{HJB}
&\max\bigg\{c V_x(x,l)+\frac{\sigma^2}2 V_{xx}(x,l)+p V_l(x,l)+\frac{\delta^2}2 V_{ll}(x,l)+\rho\delta\sigma V_{xl}(x,l)\nonumber
\\&{}+\int_{\R^2}V(x+h_2,l+h_1)-V(x,l)-V_x(x,l)h_2\one_{[|h_2|\le 1]}-V_l(x,l)h_1\one_{[|h_1|\le 1]}\;\nu(\td h_1,\td h_2),\nonumber
\\&\quad{} e^{-l}-V_x(x,l)\bigg\}=0,
\end{align}
subject to the boundary condition
\begin{equation}
\label{eq:bc}
\begin{cases}
V(x,l) =e^{-l} w(x), &  \text{for all $x<0$},\\
\vspace{-0.6cm}\\
V(0,l) = e^{-l}w(0), &  \text{in the case $\sigma^2>0$ or $\int_{-1}^0 y e^{\alpha y}\;\nu_R(\td y) = \infty$}.
\end{cases}
\end{equation}
The second part of the HJB equation \eqref{HJB}, multiplied by $e^{l}$ yields $1-e^lV_x(x,l)$, which can result in a constant barrier strategy for the surplus if $e^lV_x(x,l)$ does not depend on $l$.
Since we conjecture that the optimal strategy is of a constant barrier type, the value function should have the form $e^{-l}F(x)$ and the HJB equation becomes
\begin{align}\label{HJB2}
\max\bigg\{&c F'(x)+\frac{\sigma^2}2 F''(x)-p F(x)+\frac{\delta^2}2 F(x)-\rho\delta\sigma F'(x)
\\&{}+ \int_{\R^2}e^{-h_1}F(x+h_2)-F(x)-F'(x)h_2\one_{[|h_2|\le 1]}+F(x)h_1\one_{[|h_1|\le 1]}\;\nu(\td h_1,\td h_2),\nonumber
\\&{} 1-F'(x)\bigg\}=0\;.\nonumber
\end{align}
subject to the boundary condition
\begin{equation}\label{eq:bc2}
\begin{cases}
F(x) = w(x), &  \text{for all $x<0$},\\
\vspace{-0.6cm}\\
F(0) = w(0), &  \text{in the case $\sigma^2>0$ or$\int_{-1}^0 y e^{\alpha y}\;\nu_R(\td y) = \infty$}.
\end{cases}
\end{equation}
If $w(x)=0$ for all $x< 0$, which corresponds to Gerber-Shiu function $\mathcal{W}(\pi)=0$, then the boundary condition \eqref{eq:bc2} is equivalent to the requirement that $F$ equals zero on the negative half-line.
\bigskip
\\In order to prove the optimality of a barrier strategy we consider the HJB equation \eqref{HJB2} with boundary conditions \eqref{eq:bc2} first.
\begin{theorem}[Verification Theorem]\label{verthm}
Let $\pi$ be an admissible dividend strategy such that $V_\pi$ is
twice continuously differentiable
and ultimately dominated by some affine function.
If \eqref{HJB} -- \eqref{eq:bc} hold true for $V_\pi$ then $V_\pi(x,l)=V(x,l)$ for all $x\geq 0$, $l\in\R$.
\end{theorem}
\begin{proof}
For the sake of clarity of presentation, the proof is postponed to p.\ \pageref{proof:1} in Section \ref{sec:proof}.
\end{proof}
Now, we will focus on the set of barrier strategies paying out any excess above a given level as dividends. Let $a>0$ and $\pi_a$ denote a barrier and the corresponding barrier strategy. In the following, we will investigate the properties of the return functions corresponding to barrier strategies in order to apply Theorem \ref{verthm}.
For simplicity, we will denote the return function corresponding to the strategy $\pi_a$ by $V_a(x,l)=e^{-l}F_a(x)$, i.e.
\[
F_a(x):=e^l\mE_{(x,l)}\bigg[\int_0^{T^{\pi_a}} e^{-Y_t}\md L^{\pi_a}_t+\mathcal{W}(\pi_a)\bigg]\;.
\]
The above representation is possible because the underlying barrier does not depend on $Y$.
\begin{theorem}\label{barthm}
It holds
\begin{equation}\label{va}
F_{a}(x) =
\begin{cases} \frac{\Delta (x)}{\Delta'(a)}\big(1-\gqw^\prime(a)\big)+ \gqw(x),
& x \leq a,\\
x - a + F_a(a), & x >  a
\end{cases}
\end{equation}
with $\gqw$ defined in \eqref{Gqw}. The function $F_a$ is continuously differentiable with respect to $x$ on $[0,\infty)$.
\end{theorem}
\begin{proof}
See Section \ref{sec:proof}.
\end{proof}
Let
\[\hp(y):= \frac{1-\gqw^\prime(y)}{\Delta' (y)}\]
and define a candidate for the optimal dividend barrier by
\begin{equation*}
 a^*:=\sup\left\{a\geq0:\hp(a)\geq\hp(x)\text{ for all }x\geq0\right\},
\end{equation*}
where $\hp(0)=\lim\limits_{x\downarrow 0}\hp(x)$.\medskip
\\
Now, using the above two theorems we can give necessary and sufficient conditions
for the barrier strategy to be optimal.
\begin{theorem}\label{ver2}
The value function $V_{a^*}(x,l)=e^{-l}F_{a^*}(x)$ under the barrier strategy $\pi_{a^*}$ is in the domain of the full generator $\mathfrak A$.
The barrier strategy $\pi_{a^*}$ is optimal and $V_{a^*}(x,l)=V(x,l)$ for all $x\geq 0$ and $l\in\R$ if and only if
\begin{equation}\label{inq2}
 \mathfrak A [e^{-l}F_{a^*}(x)]\leq 0\quad \text{for all }x>a^*.
\end{equation}
 \end{theorem}
\begin{proof}
See p.\ \pageref{proof:2} in Section \ref{sec:proof}.
\end{proof}
\begin{theorem}\label{suffcond}
Suppose that
\begin{equation}\label{as1}
 \hp(a)\geq\hp(b)\quad\text{for all } a^*\leq a\leq b.
\end{equation}
Then the barrier strategy with the barrier $a^*$ is the optimal strategy, that is, $V(x,l)=e^{-l}F_{a^*}(x)$ for all $x\geq 0$.
\end{theorem}
\begin{proof}
See p.\ \pageref{proof:3} in Section \ref{sec:proof}.
\end{proof}
\begin{corollary}\label{cor:explopt}
Assume that $w(x)=0$ (there is no penalty function) and that \eqref{eq:condW} holds true.
Then $\pi_{a^*}$ is the optimal strategy.
\end{corollary}
\begin{proof}
See p.\ \pageref{proof:4} in Section \ref{sec:proof}.
\end{proof}
\begin{remark}\label{Uwaga}\rm
If $w(x)=0$ for $x\leq 0$, that is, there is no penalty function, then
under the assumption that $f^\alpha$ is monotone decreasing, we have
\begin{equation}\label{valuefunctionmain}
V(x,l)=V_{a^*}(x,l)=e^{-l}F_{a^*}(x)=e^{-l}\cdot \frac{\Delta(x)}{\Delta'(a^*)}\;,\end{equation}
where $a^*$ maximises $H^\prime_\alpha(x)=1/\Delta'(x)$ hence solves
\[\Delta^{''}(a^*)=0\]
which is equivalent to the requirement that
\begin{equation}\label{second}
\frac{\md^2}{\md x^2}V(a^*,l)=0\;.\end{equation}
In other words, knowing the barrier strategy is optimal, identifying the value function
could be based on solving HJB equation \eqref{HJB2} (without any boundary conditions) and
finding $a^*$ via \eqref{second} and using the boundary condition $\frac{\md }{\md x}V(a^*,l)=e^{-l}$ or equivalently $F_{a^*}'(a^*)=1$.
\end{remark}
\section{Examples}\label{sec:examples}
In this section we pick up the idea of continuous dependence and flash crashes on the global market impacting both the exchange rate and the surplus of an insurance company. In the first example below we deal with the continuous dependency case, while Example 2 considers the flash crashes. By assuming that the jumps in the considered L\'evy processes are exponentially distributed we are able to rewrite the HJB equation in terms of an ordinary differential equation of order 3. In this case, we can show that the problem of finding the optimal barrier and the value function transforms in solving the underlying differential equation with corresponding boundary conditions.
\begin{example}\label{firstone}\rm
Let us first consider the following example. The classical model of risk theory describes the surplus of an insurance entity up to infinity.
We let $N_t$ be the jump number Poisson process with intensity $\lambda$, $c$ the premium rate, $C_i$ iid claim sizes Exp$(\gamma)$-distributed. We let the surplus be given by the perturbed classical risk model and the exponential expression of the exchange rate by a Brownian motion with drift.
\begin{align*}
&R_t=x+c t -\sum_{i=1}^{N_t}C_i+\sigma B_t\quad \mbox{and}\quad Y_t=l+pt+\delta W_t\;,
\end{align*}
where $B$ and $W$ are Brownian motions with correlation coefficient $\rho$. Further, in order to guarantee the well-posedness of our problem we assume $p>\frac{\delta^2}2$, confer Assumption \eqref{assump:2}.
\\In the following, we first derive the value function directly from the HJB equation and show in the second part the derivation of the value function and the optimal barrier via scale functions.
\medskip
\subsubsection*{Derivation of the value function via HJB.}
The HJB equation \eqref{HJB2} has the following form
\begin{align*}
\max\bigg\{c F'(x)+ \frac{\sigma^2}2 F''(x) +\lambda\int_{0}^x F(x-y)\md G(y)-\big(\lambda+p-\frac{\delta^2}2\big)& F(x)-\rho\delta\sigma F'(x),
\\&1-F'(x)\bigg\}=0\;,
\end{align*}
where $G(y)=1-e^{-\gamma y}$. Let $g(x):=\int_0^x F(y)e^{\gamma y}\md y$. Then
\begin{align*}
\lambda\gamma g(x)+\Big(\frac{\delta^2+\sigma^2\gamma^2}2-p-\lambda-\gamma c+\rho\gamma\delta\sigma\Big)g'(x)+\big(c-\delta\rho\sigma-\gamma &\sigma^2\big)g''(x)
\\&{}+\frac{\sigma^2}2 g'''(x)= 0\;.
\end{align*}
Let for the sake of clarity
\begin{align*}
&a_2:=\frac2{\sigma^2}\big(c-\delta\rho\sigma-\gamma\sigma^2\big),
\\&a_1:=\frac2{\sigma^2}\big(p+\frac{\delta^2}2-\lambda+\frac{\sigma^2\gamma^2}2-\gamma c+\rho\gamma\delta\sigma\big),
\\&a_0:=\frac{2\lambda\gamma}{\sigma^2}.
\end{align*}
Define
\[
P(s):=s^3+a_2s^2+a_1s+a_0.
\]
If $s_i$, $1\le i\le n$ are different zeros of $P(s)$ and $\lambda_i$ $1\le i\le n$ the corresponding multiplicities with $n\le 3$, then
due to \cite[p.\ 105]{kamke} or \cite{walter}, all solutions to the above differential equation are given by
\[
e^{s_1 x}P_{\lambda_1-1}(x)+...+e^{s_n x}P_{\lambda_n-1}(x)
\]
where $P_h$ is a polynomial of the degree $\le h$. Concerning the zeros of $P(s)$, we can distinguish between 2 cases: $P(s)$ has 3 real zeros, $P(s)$ has 1 real and 2 complex zeros (complex conjugates). In the second case, the general solution is $e^{s_1x}C_1+e^{s_2x}\sin(x)C_2+e^{s_2x}\cos(x)C_3$.\medskip
\\Considering again the equation
\[
c F'(x)+ \frac{\sigma^2}2 F''(x) +\lambda\int_{0}^x F(x-y)\md G(y)-\big(\lambda+p-\frac{\delta^2}2\big) F(x)-\rho\delta\sigma F'(x)=0\;,
\]
yields $F''(x)>0$ if $F'(x)=0$ and $F'''(x)>0$ if $F''(x)=0$ and $F'(x)>0$. Therefore, for an $a^*$ fulfilling $F''(a^*)=0$ and $F'(a^*)=1$ we have $F'(x)>1$ on $[0,a^*)$, i.e.\ $F$ fulfils the HJB equation.
To identify the optimal level $a^*$ note that by Remark \ref{Uwaga}
the boundary conditions are given by the following equations: $g(0)=0$, $g'(0)=0$, $F'(a^*)=1$ and $F''(a^*)=0$ for some $a^*$.
\begin{figure}[t]
\includegraphics[scale= 0.35, bb = 20 0 500 550]{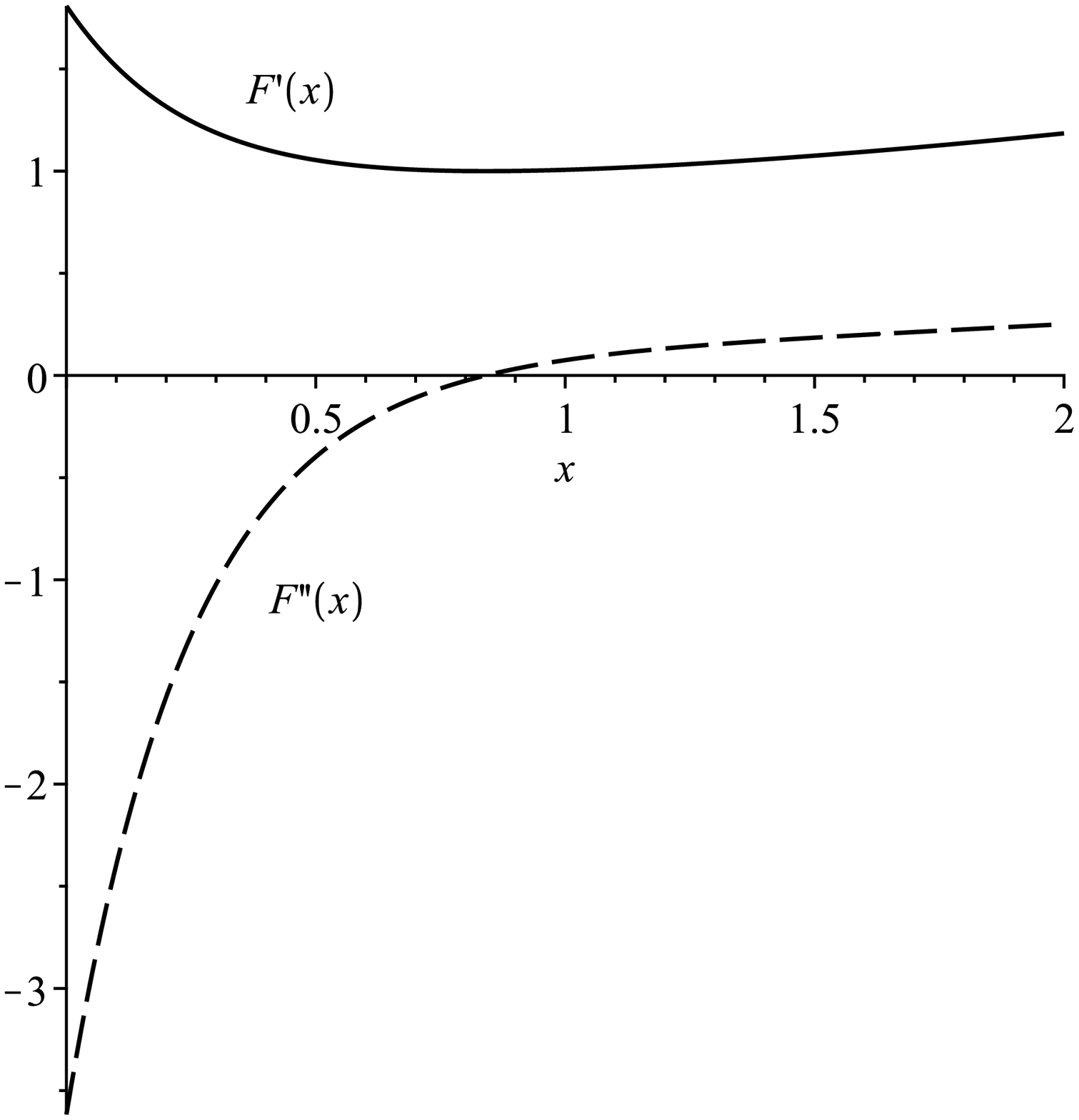}
\includegraphics[scale = 0.35, bb = -110 0 500 550]{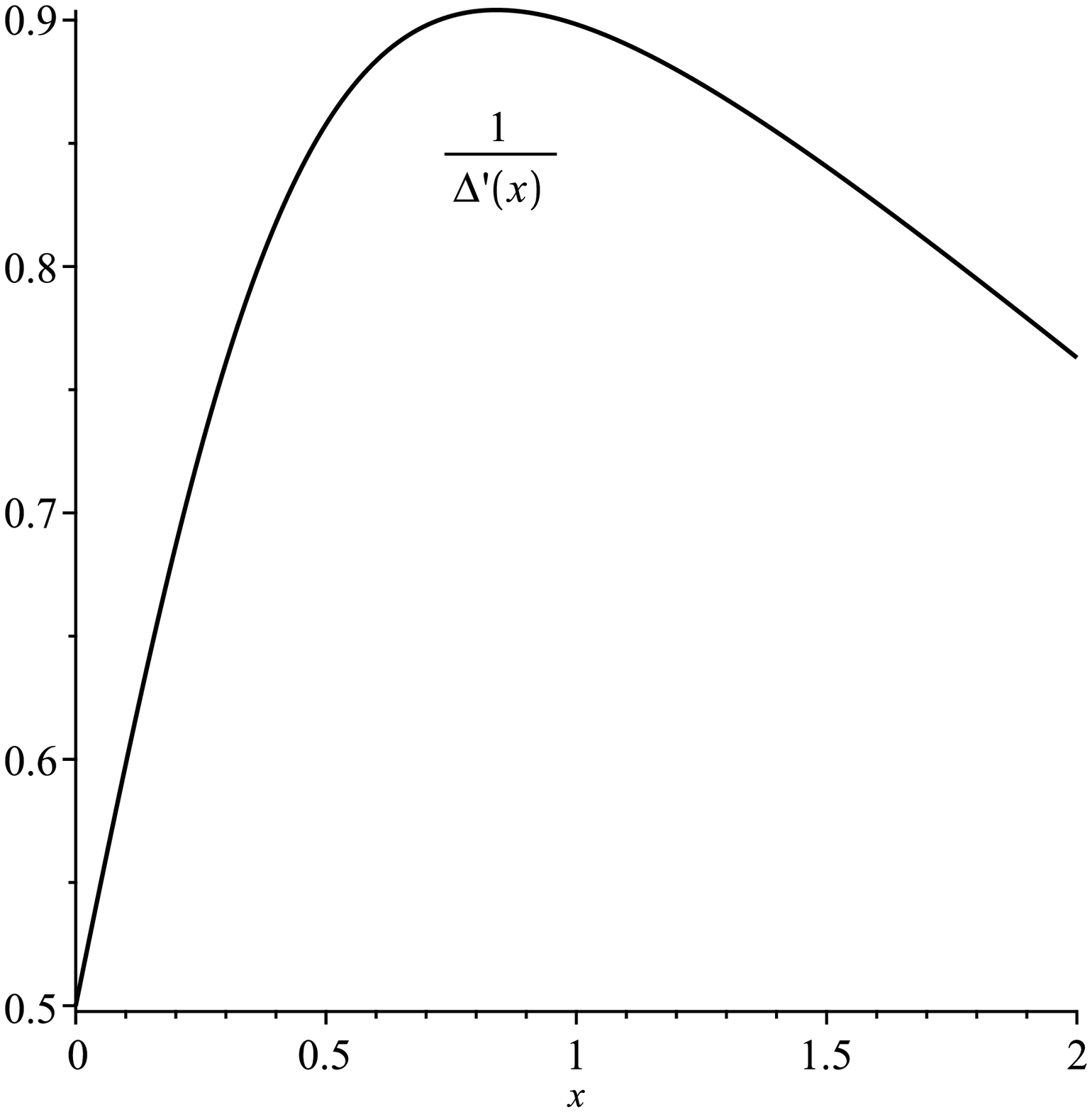}
\caption{The derivatives $F'(x)$ and $F''(x)$ on the interval $[0,2]$ (left picture) and $\frac 1{\Delta'(x)}$ (right picture) with the optimal barrier $a^*=0.840599$. \label{fig:1}}
\end{figure}
\smallskip
\\For instance, for $c= 1.3$, $p=0.6$, $\sigma = 1$, $\delta= 1$, $\rho =0.3$, $\lambda= 2$ and $\gamma= 2$. Then,
\[
g(x)=C_1e^{s_1 x}+C_2e^{s_2 x}+C_3e^{s_3 x},
\]
where $s_1= 1.697007$, $s_2= 2.327991$ and $s_3= -2.024999$. The boundary conditions yield the unique solution, (confer \cite[p.\ 199]{walter})
\[
g(x)=\Big(-6.898735\cdot e^{s_1x}+5.898735\cdot e^{s_2x}+e^{s_3x}\Big)0.111605\;.
\]
The solution to the HJB equation \eqref{HJB2}, $F(x)$, is then given by
\[
F(x)=g'(x)e^{-\gamma x}=-1.306593 e^{-0.302992x}+1.532594 e^{0.327991x}-0.226002 e^{-4.024999x}\;.
\]
The boundary conditions yield the optimal dividend barrier $a^*=0.840599$.
Figure \eqref{fig:1}, left picture, illustrates the first and the second derivatives of the value function $F(x)$, where we see that $F'(x)>1$ and $F''(x)<0$ for $x\in[0,a^*)$ and $F'(a^*)=1$, $F''(a^*)=0$.
\\Since the function $e^{-l}F(x)$ is twice continuously differentiable with respect to $x$, it can be shown using the standard methods, confer for instance \cite{Schmbook} that $e^{-l}F(x)$ is the value function.
\subsubsection*{Derivation via scale functions.}
Coming from the other side, using Theorem \ref{barthm} we can derive the value function and the optimal barrier via scale functions. First of all, find $\alpha \geq 0$ which sets the Laplace exponent of the bivariate L\'evy process $(R,Y)$ to zero:
\[
\psi(\alpha,-1)=(c-\rho\delta\sigma)\alpha+\frac{\sigma^2\alpha^2 }{2}+\lambda \Big(\frac{\gamma}{\gamma+\alpha}-1\Big)+\frac{\delta^2}{2}-p =0\;.
\]
Having identified $\alpha=0.32799143$ one can calculate the function $\psi_R(\beta)$ due to \eqref{tildepsi} (or by \eqref{newtriple}):
\[
\psi_R(\beta)=\psi(\beta+\alpha,-1)=(c+\alpha \sigma^2-\rho\delta\sigma)\beta+\frac{\sigma^2\beta^2}2-\frac{\lambda\gamma}{\gamma+\alpha}\cdot\frac{\beta}{\gamma+\alpha+\beta}\;.
\]
Now, using \eqref{eq:defW}, we can get $W^\alpha(x)$. Noting that the zeros of $\psi_R$ are given by $\tilde s_1=0$, $\tilde s_2=-0.630984$, $\tilde s_3=-4.352991$ and using the inverse Laplace transform, we get
\[
W^\alpha(x)= -0.249971 e^{-4.352991 x}-1.445168 e^{-0.630984 x}+1.695139\;.
\]
Therefore, we can conclude
\[
\Delta(x)=e^{\alpha x}W^\alpha(x)= -0.249971 e^{-4.024999x}-1.445168 e^{-0.302992x}+1.695139 e^{0.327991x}\;.
\]
By \eqref{newtriple} and the form of $\psi_R$ given above, the density $f^\alpha(y)=(\gamma+\alpha)e^{-(\gamma+\alpha) y}$ of the generic jump size $C$ of the surplus under $\Prob^{(\alpha,-1)}$ is completely monotone. 
Hence, from Corollary \ref{cor:explopt} the barrier strategy $\pi_{a^*}$ is optimal.
Due to \eqref{va}, the value function and the optimal strategy are given by
\begin{align*}
&V(x,l)=\begin{cases}
e^{-l}\frac{\Delta(x)}{\Delta'(a^*)}& \mbox{: $x\le a^*$}\;,\\
V(a^*,l)+x-a^* & \mbox{: $x> a^*$}\;,
\end{cases}
\\&a^*=\sup\left\{a\geq0:\frac1{\Delta'(a)}\geq \frac1{\Delta'(x)}\text{ for all }x\geq0\right\}=0.840599\;.
\end{align*}
In Figure \ref{fig:1}, the left picture illustrates that $\frac1{\Delta'(x)}$ has the global maximum at $0.840599$.
Since all assumptions of Corollary \ref{cor:explopt} are satisfied, the value function is given in \eqref{valuefunctionmain}
and it is consistent with the previous analysis.
\end{example}
\begin{example}\rm
In this example, we again assume that the surplus process of the considered insurance company $R_t$ is given by a perturbed classical risk model and the exponential of the exchange rate $Y_t$ by a continuous drift and a jump part, where the number of jumps is correlated with the number of jumps in the surplus.
Let
\begin{align*}
& R_t=x+c t +\sigma B_t-\s_{i=1}^{N_t}C_i \quad\mbox{and}\quad Y_t=l+p t-\s_{i=1}^{M_t} Z_i\;,
\end{align*}
where $B_t$ is a standard Brownian motion, $C_i$ describe the jumps in the surplus and $Z_i$ jumps in the exchange rate, where the sequences $(C_i)_{i\ge 1}$ and $(Z_i)_{i\ge 1}$ are independent. As we assume that the crashes are not severe, we let $C_i$ have the distribution function $G(x)=1-e^{-\gamma x}$ and the distribution function of $Z_i$ is $H(x)=1-e^{-\eta x}$, i.e. the jumps are not heavy-tailed. Let further $\bar N_t$ be a Poisson process with parameter $\bar\lambda$ independent of the Poisson process $M_t$ with parameter $\theta$. We let $N_t=\bar N_t+M_t$, i.e. $N_t$ is again a Poisson process with parameter $\lambda=\bar\lambda+ \theta$.
\\In order for the problem to be well-defined, Assumption \eqref{assump:2}, we require
\[
\psi(0,-1)=-p+\theta \frac{\eta}{\eta-1}-\theta<0\;.
\]
\subsubsection*{Derivation of the value function via HJB.}
In this  case, HJB equation \eqref{HJB2} (divided by $e^{-1}$) has the form
\begin{align*}
\max\Big\{cF'(x)+\frac{\sigma^2}2 F''(x)-p F(x)
+\int_{\R^2}e^{-h_2}F(x+h_1)-F(x)\;&\nu(\md h_1,\md h_2),
 1-F'(x)\Big\}=0\;.
\end{align*}
The integral in the above equation can be written as follows
\begin{align*}
\int_{\R^2}e^{-h_2}F(x+h_1)-F(x)\;\nu(\md h_1,\md h_2)&=\theta\int_0^\infty \int_0^x e^{z} F(x-y)\md G(y)\md H(z)
\\&\quad {}+\bar \lambda\int_0^x  F(x-y)\md G(y)-\lambda F(x)
\\&=\frac{\theta \gamma \eta e^{-\gamma x}}{\eta-1}\int_0^x F(y)e^{\gamma y}\md y
\\&\quad {}+\bar \lambda\gamma e^{-\gamma x}\int_0^x  F(y)e^{\gamma y}\md y-\lambda F(x)\;.
\end{align*}
\begin{figure}[t]
\includegraphics[scale= 0.35, bb = 20 0 500 500]{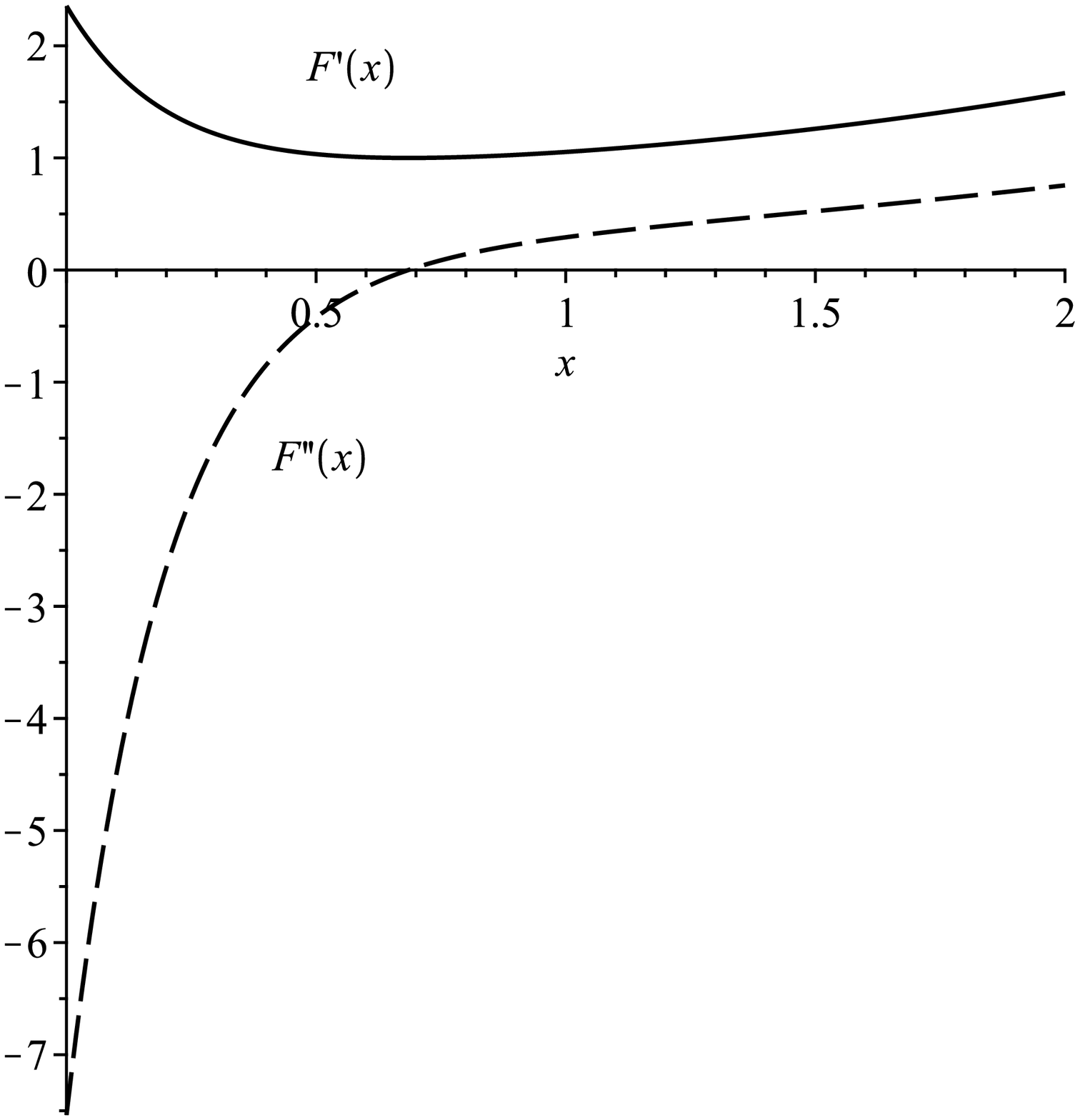}
\includegraphics[scale = 0.35, bb = -120 0 500 500]{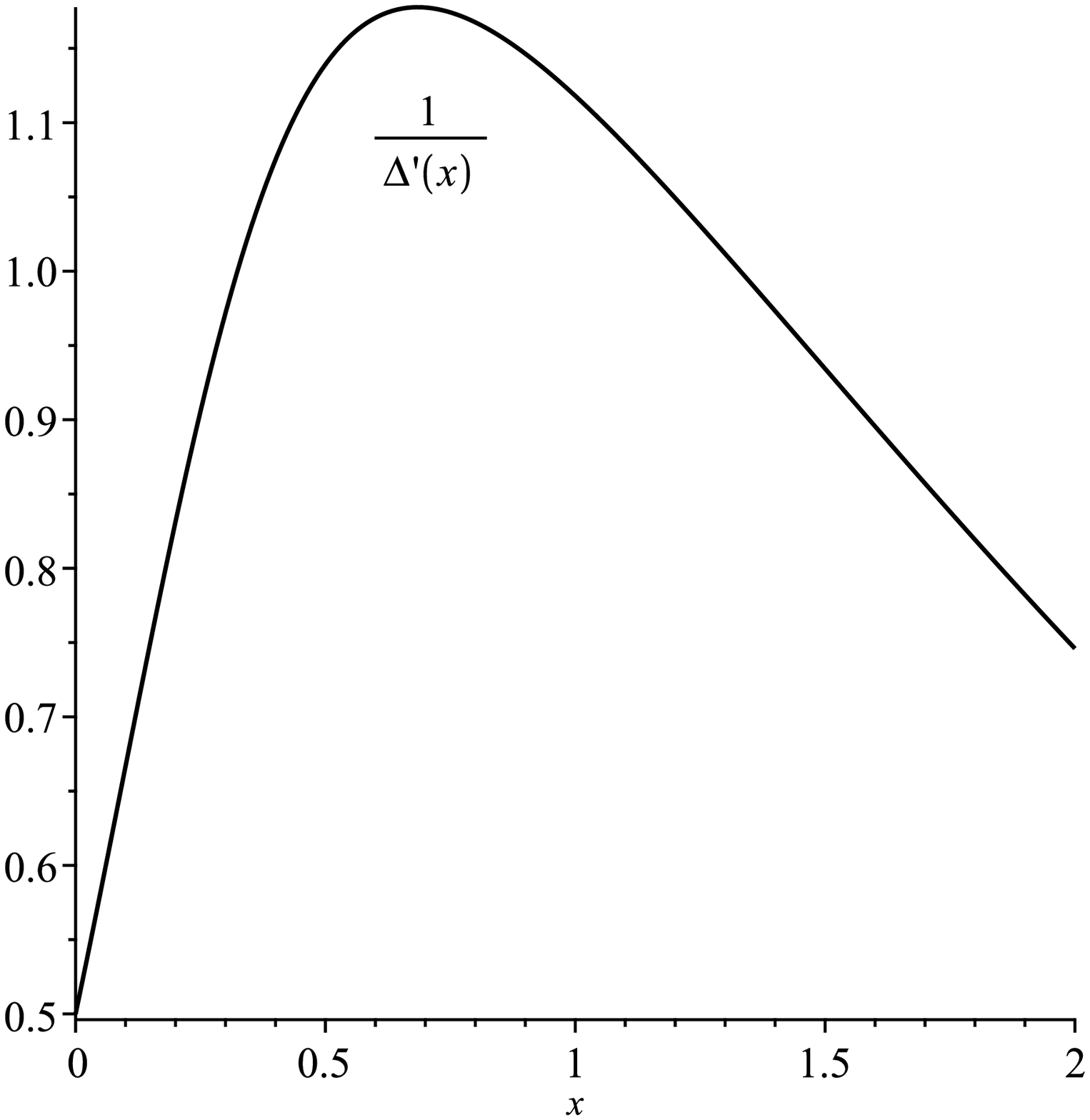}
\caption{The derivatives $F'(x)$ and $F''(x)$ on the interval $[0,2]$ (left picture) and $\frac 1{\Delta'(x)}$ (right picture) with the optimal barrier $a^*=0.684809$. \label{fig:2}}
\end{figure}
Like in the previous example, consider the differential equation with \\$g(x)=\int_0^x  F(y)e^{\gamma y}\md y$:
\begin{align*}
\Big(\frac{\theta\gamma \eta}{\eta-1}+\bar\lambda\gamma\Big) g(x)+\Big(\frac{\sigma^2\gamma^2}2-p-\lambda-\gamma c\Big)g'(x)+\big(c-\gamma\sigma^2\big)g''(x)+\frac{\sigma^2}2 g'''(x)= 0&\;.
\end{align*}
Let now $c = 1.6$, $p=0.6$, $\sigma= 1$, $\gamma = 2$, $\theta = 2$, $\bar \lambda := 0.5$, $\eta= 5$.
\\The solution is given by
\[
g(x)=\Big(-4.598667e^{s_1x}+3.598667e^{s_2x}+e^{s_3x}\Big)0.09447
\]
with $s_1=1.329911$, $s_2=2.557360$, $s_3=-3.087271$. The solution to the HJB is then given by
\[
F(x)=\Big(-4.598667s_1e^{(s_1-2)x}+3.598667s_2e^{(s_2-2)x}+s_3e^{(s_3-2)x}\Big)0.09447\;.
\]
The boundary conditions yield $a^*=0.684809$. Figure \ref{fig:2}, left picture, illustrates the derivatives $F'$ and $F''$ on the interval $[0,2]$: $F'(x)>1$ and $F''(x)<0$ on $[0, 0.684809)$.
\\Since the solution to the HJB, $e^{-1}F(x)$ is twice continuously differentiable with respect to $x$, using Ito's formula, confer \cite{Schmbook}, one can prove that $e^{-1}F(x)$ is indeed the value function.
\subsubsection*{Derivation via scale functions.}
Consider first the Laplace exponent of the bivariate L\'evy process $(R,Y)$. Find $\alpha \geq 0$ setting the Laplace exponent to zero:
\[
\psi(\alpha,-1)=\frac{\sigma^2\alpha^2 }{2}+c\alpha+\theta \Big(\frac{\gamma}{\gamma+\alpha}\cdot \frac{\eta}{\eta-1}-1\Big)+\bar\lambda \Big(\frac{\gamma}{\gamma+\alpha}-1\Big)-p =0\;.
\]
It holds $\alpha=0.557360$, leading to
\[
\psi_R(\beta)=\psi(\beta+\alpha,-1)=\frac{\sigma^2\beta^2 }{2}+\beta\big(c+\sigma^2\alpha\big)+\frac\gamma{\gamma+\alpha}\Big(\frac{\theta \eta}{\eta-1}+\bar\lambda\Big)\Big(\frac{\gamma+\alpha}{\gamma+\alpha+\beta}-1\Big)\;.
\]
Using inverse Laplace transform one gets due to \eqref{eq:defW} and Lemma \ref{Wrepr}
\[
\Delta(x)=e^{\alpha x}W^\alpha(x)= e^{0.557360x}\Big(-0.2476417475e^{-5.644632x}-0.490573e^{-1.22745x}+0.738215\Big)\;.
\]
From the representation of $\psi_R(\beta)$ above, we find that the density of the jumps in the surplus under the measure $\mP^{(\alpha,-1)}$ is given by $f^\alpha(x)=(\alpha+\gamma)e^{-(\alpha+\gamma)x}$. Since $f^\alpha$ is completely monotone, by Corollary \ref{cor:explopt} and Remark \ref{Uwaga} the optimal strategy is of barrier type.
\\The value function and the optimal strategy are
\begin{align*}
&V(x,l)=\begin{cases}
e^{-l}\frac{\Delta(x)}{\Delta'(a^*)}& \mbox{: $x\le a^*$}\;,\\
V(a^*,l)+x-a^* & \mbox{: $x> a^*$}\;,
\end{cases}
\\&a^*=\sup\left\{a\geq0:\frac1{\Delta'(a)}\geq \frac1{\Delta'(x)}\text{ for all }x\geq0\right\}=0.684809\;.
\end{align*}
The function $\frac1{\Delta'(x)}$ is illustrated in Figure \ref{fig:2}, right picture. And the achieved results are in line with the results derived via solving the HJB equation directly.
\end{example}
\section{Proofs}\label{sec:proof}
\subsection{Proof of Verification Theorem \ref{verthm}\label{proof:1}}
The proof is based on a
representation of $v$ as the pointwise minimum of a class of
``controlled'' supersolutions to the HJB equation.
We start with the observation that the value function satisfies the following dynamic programming equation.
\begin{lemma}\label{dynamic}
After extending $V$ to the negative half-axis by $V(x)=w(x)$ for $x<0$, we have,
for any stopping time $\tau$,
$$V(x,l)=\sup_{\pi\in \Pi} \E_x\bigg[  V(R^\pi_{\tau\wedge T},Y_{\tau\wedge T})+\int_0^{\tau\wedge T}e^{-Y_t} \,dL_t^\pi\bigg].
$$
\end{lemma}
\begin{proof}
This follows by a straightforward adaptation of  classical arguments
(see  e.g. \cite[pp.~276--277]{AM}).
We will prove that $v$ is a supersolution to HJB equation \eqref{HJB}.
\end{proof}
\begin{lemma}\label{eq:Vpilemma}
The process
\begin{equation}\label{eq:Vpi}
V^\pi_t := V(R^\pi_{t\wedge T},Y_{t\wedge T}) +
\int_0^{t\wedge T} e^{-Y_s}\, dL^\pi_s
\end{equation}
is a uniformly integrable (UI) supermartingale.
\end{lemma}
\begin{proof}
Fix arbitrary $\pi\in\Pi$, $x\geq 0$ and $s,t\geq 0$
with $s<t$.
The process $V^\pi_t$ is $\mathcal{F}_t$-measurable,
and is UI. Indeed, by Lemma \ref{dynamic} we have
\begin{equation*}
\E_{(x,l)}[V^\pi_t]\leq  \sup_{\pi\in\Pi}\E_{(x,l)}\bigg[ V(R^\pi_{t\wedge T},Y_{t\wedge T}) +
\int_0^{t\wedge T} e^{-Y_s} dL^\pi_s\bigg]=V(x,l).
\end{equation*}
Now by integration by parts, the non-positivity of $w$
and the no exogenous ruin assumption
\begin{equation}
V(x,l) \leq (A x+B)e^{-l},
\label{affinedomination}
\end{equation}
for some constants $A,B>0$.
\smallskip
\\Let $W^\pi_t$ be the following value process:
\begin{align}
& W_s^\pi := \operatornamewithlimits{ess\,sup}_{\tilde\pi\in\Pi_s}
J_s^{\tilde\pi}, \qquad J_s^{\tilde\pi} := \E\bigg[
\int_0^{T^{\tilde\pi}} e^{-Y_u}\,dL^{\tilde\pi}_u +
e^{-Y_{T^{\tilde\pi}}}w(R^{\tilde\pi}_{T^{\tilde\pi}})\bigg|\mathcal{F}_s\bigg], \label{eq:Wspi}\\
& \Pi_s: = \left\{\tilde\pi = (\pi,\overline{\pi}) = \{L^{\pi,\overline{\pi}}_u, u\geq 0\}:
\overline{\pi}\in\Pi\right\},\qquad
L^{\pi,\overline{\pi}}_u: =
\begin{cases}
L_u^\pi, & u\in[0,s[,\\
L_s^\pi + L_{u-s}^{\overline{\pi}}(R^\pi_s), & u\ge s,
\end{cases}\notag
\end{align}
where $L^{\overline{\pi}}(x)$ denotes the process of cumulative dividends of the strategy
$\overline{\pi}$ corresponding to the initial capital $x$.
\smallskip
\\
The fact that $V^\pi$ is a supermartingale is a direct
consequence of the following $\Prob$-a.s. relations:
\begin{itemize}
\item[(a)] $V_s^\pi = W_s^\pi$, \quad (b) $W_s^\pi\ge \E[W^\pi_t|\mathcal{F}_s]$, where $W^\pi$ is the process defined in (\ref{eq:Wspi}).
\end{itemize}

Point (b) follows by classical arguments,
since the family $\{J_t^{\tilde\pi}, \tilde\pi\in\Pi_t\}$ of random variables
is  upwards directed; see Neveu \cite{Neveu} and Avram et al. \cite[Lem. 3.1(ii)]{APP2} for details.

To prove (a), note that on account of
the Markov property of $R^\pi$ and $Y_t$ it also follows that
conditional on $R^\pi_s$,
$\{R^{\tilde\pi}_u-R^{\tilde\pi}_s, u\ge s \}$
is independent of $\mathcal{F}_s$. As a consequence,
the following identity holds on the set $\{s < T^{\tilde\pi}\}$:
\begin{multline*}
\E\bigg[\int_0^{T^{\tilde\pi}} e^{-Y_u}\,dL^{\tilde\pi}_u +
e^{-Y_{T^{\tilde\pi}}}
w(R^{\tilde\pi}_{T^{\tilde\pi}})\bigg|\mathcal{F}_s\bigg]\\
\begin{aligned}
&=
\E_{(R^{\pi}_s, Y_s)}\bigg[
\int_0^{T^{\overline{\pi}}} e^{-Y_u}\,dL^{\overline{\pi}}_u+
e^{-Y_{T^{\overline{\pi}}}}w(R^{\overline{\pi}}_{T^{\overline{\pi}}})\bigg]
+
\int_0^s e^{-Y_u}\,dL^{\pi}_s\\
&= V_{\overline{\pi}}(R^\pi_s,Y_s) + \int_0^s e^{-Y_u}\,dL^{\pi}_u,
\end{aligned}
\end{multline*}
and then we have the following representation:
$$
J_s^{\tilde\pi} =
V_{\overline{\pi}}(R^\pi_{s\wedge T},Y_{s\wedge T})
+ \int_0^{s\wedge T}e^{-Y_u}\,dL^{\pi}_u,
$$
which completes the proof on taking the essential supremum over the relevant family of strategies.
\end{proof}
To prove that the value function $V$ is a solution to the HJB equation \eqref{HJB},
we will denote by $\mathcal{G}$ the family of functions $g$  for which
\begin{equation}\label{eq:ggmartI}
M^{g,T_I}
:=\{g(R_{t\wedge T_I},Y_{t\wedge T_I}),\;
t\geq 0\}, \quad T_I := \inf\{t\ge0: R_t\notin I\},
\end{equation}
is a supermartingale for any closed interval $I\subset [0,\infty )$,
and such that
\begin{equation}\label{diffineq}
\frac{g (x,l)-g(y,l)}{x-y}\geq e^{-l}\quad \text{ for all }x>y \geq 0, \qquad g(x,l)\geq e^{-l}w(x) \quad\text{ for } x<0
\end{equation} and $g$ is ultimately dominated by some linear function.
\begin{lemma}\label{vHJB}
We have $V\in\mathcal{G}$.
\end{lemma}
\begin{proof}
Taking a strategy of not paying any dividends, by Lemma \ref{eq:Vpilemma} we find that
the process (\ref{eq:ggmartI}) with $g=v$ is a supermartingale.
We will show now that
\begin{equation*}
V(x,l)-V(y,l)\geq e^{-l}(x-y)\quad \text{for all $x>y\geq 0$ and $l\in\R$}.
\end{equation*}
Denote by
$\pi^\epsilon(y)$ an $\epsilon$-optimal strategy for the case $R^\pi_0=y$.
Then we take the strategy of paying $x-y$ immediately and
subsequently following the strategy $\pi^\epsilon(y)$ (note that such a strategy is admissible), so that the following holds:
$$
V(x,l) \ge (x- y)e^{-l} + V_{\pi^\epsilon}(y,l) \ge (x- y)e^{-l}  +V(y,l) - \epsilon\;.
$$
Since this inequality holds for any $\epsilon>0$, the stated lower
bound follows.
Linear domination of $v$ in $x$ by some affine function in $x$ follows from (\ref{affinedomination}).
\end{proof}
We now give the dual representations of the value function on a closed interval $I$.
Assume that $\mathcal{H}_I$ is a family of functions $k$
for which
\begin{equation*}
\WT M_t^{k,\pi}:=e^{-Y_{t\wedge \tau^\pi_I}}k(R^\pi_{t\wedge \tau^\pi_I},Y_{t\wedge \tau^\pi_I}) +
\int_0^{t\wedge \tau^\pi_I} e^{-Y_s} dL^\pi_s
\end{equation*} is an UI supermartingale
for $\tau^\pi_I:=\inf\{t\geq 0: R^\pi_t\notin I\}$
and $$k(x,l)\geq V(x,l)\qquad \text{ for }x\notin I.$$
Then
\begin{equation}
\label{repr1}
V(x,l) = \min_{k\in\mathcal{H}_I} k(x,l)\qquad\text{ for }x\in I.
\end{equation}
Indeed, let $\pi\in \Pi$, $k\in \mathcal{H}_I$ and $x\in I$.
Then the Optional Stopping Theorem applied to the UI Dynkin martingale yields
\begin{align*}
k(x,l) &\geq \lim_{t\to\infty}\E_{(x,l)}\bigg [ e^{-Y_{\tau_I^\pi\wedge
t}}k(R^\pi_{t\wedge\tau_I^\pi},Y_{t\wedge\tau_I^\pi}) + \int_0^{t\wedge\tau_I^\pi}e^{-Y_s}\,dL^\pi(s)  \bigg]\\
&\geq \E_{(x,l)}\bigg[e^{-Y_{\tau_I^\pi}}V(R^\pi_{\tau_I^\pi},Y_{\tau^\pi_I}) +
\int_0^{\tau_I^\pi}e^{-Y_s}\,dL^\pi(s)      \bigg],
\end{align*}
where the convention $\exp\{-\infty\}=0$ is used.\\
Taking the supremum over all $\pi\in\Pi$ shows that $k(x,l)\geq V(x,l)$.
Since $k\in\mathcal{H}_I$ was arbitrary, it follows that
$$\inf_{k\in\mathcal{H}_I}k(x,l)\geq V(x,l).
$$
This inequality is in fact an equality since $V$ is a
member of $\mathcal{H}_I$ by Lemma~\ref{eq:Vpilemma}.
The value function $V$ admits a more important representation from which the Verification Theorem
\ref{verthm} follows.
\begin{proposition}\label{repr2}
We have
\begin{equation*}
V(x,l) = \min_{g\in\mathcal G}g(x,l).
\end{equation*}
\end{proposition}
\begin{proof}
Since $v\in\mathcal{G}$ in view of Lemma \ref{vHJB}, by (\ref{repr1}) it suffices to prove that $\mathcal{G}\subset \mathcal{H}_{[0,\infty)}$.
The proof of this fact is  similar to the proof of the shifting lemma \cite[Lem. 5.5]{APP2}.
For completeness, we give  the main steps.
Fix arbitrary $g\in \mathcal{G}$, $\pi\in\Pi$ and $s,t\geq 0$ with $s< t$.
Note that $\WT M^{g,\pi}$ is adapted
and UI by the linear growth condition and arguments  in the proof of Lemma \ref{eq:Vpilemma} and by \cite[Sec. 8]{APP2}.
Furthermore, the following (in)equalities hold true:
\begin{equation*}
\E\big[\WT M_t^{g,\pi}\big|\mathcal{F}_{s\wedge T}\big] \stackrel{(a)}{=} \lim_{n\to\infty}
\E\big[\WT M^{g,\pi_n}_t\big|\mathcal{F}_{s\wedge T}\big] \stackrel{(b)}{\leq} \lim_{n\to\infty}
\WT M^{g,\pi_n}_{s\wedge T} \stackrel{(c)}{=}
\WT M_{s\wedge T}^{g,\pi}
\stackrel{(d)}{=}
\WT M^{g,\pi}_s,
\end{equation*}
where the sequence $(\pi_n)_{n\in \mathbb{N}}$ of strategies is defined
by $\pi_n=\{L^{\pi_n}_t, t\geq 0\}$ with $L_0^{\pi_n} = L_0^\pi$ and
\begin{align*}
L^{\pi_n}_u &:=
\begin{cases}
\sup\{L^\pi_{v}: v< u, v\in \mathbb{T}_n\}, & 0<u < T, \\
L^{\pi_n}_{T-}, & u\ge T,
\end{cases}\\
 \mathbb{T}_n &:=\left(\left\{t_k:=s+(t-s)\frac{k}{2^n}, k\in
\mathbb{Z}\right\}\cup\{0\}\right)\cap \mathbb{R}_+,
\end{align*}
where the above $T$ is calculated for the strategy $\pi$.
Since $s$ and $t$ are arbitrary,
it  follows that $\WT M^{g,\pi}$ is a supermartingale, which will complete the proof.

Points (a), (c) and (d) follow from the Monotone and Dominated Convergence Theorems.
To prove  (b), let $T_i:=T\wedge t_i$,
 denote $ \WT M^{g,\pi_n}=M$,
$ L^{\pi_n}=L$ and observe that
\begin{align*}
&M_t - M_s =  \sum_{i=1}^{2^n} Q_i + \sum_{i=1}^{2^n}{Z_i},\quad\text{with}\\
&Q_i := g\left(R^{}_{T_i-},Y_{T_i}\right) - g(R^{}_{T_{i-1}},Y_{T_{i-1}}),\\
&Z_i:=\left(g(R^{}_{T_i},Y_{T_i}) - g(R^{}_{T_i-} ,Y_{T_i}) + \Delta L_{T_i}\right)\one_{[\Delta L_{T_i}>0]}.
\end{align*}
The strong Markov property of $R$ and $Y$ and
the  definition of $R^\pi$ imply
\begin{align}
  \E\big[
g(R^{}_{T_{i-}},Y_{T_i}) - g(R^{}_{T_{i-1}},Y_{T_{i-1}}) \big|\mathcal{F}_{T_{i-1}}\big]
=
\E_{(R_{T_{i-1}},Y_{T_{i-1}})}[g(R^{}_{\tau_i},Y_{\tau_i}) -
g(R^{}_0,Y_0)],  \label{eq:YEDOOB}
\end{align}
with $\tau_i := T_i\circ\theta_{T_{i-1}}$, where $\theta$ denotes the shift operator.
The right-hand side of \eqref{eq:YEDOOB}
is  non-positive because $g\in \mathcal{G}$.
Furthermore, it follows from \eqref{diffineq}
 that all the $Z_i$ are
non-positive. The tower property of conditional
expectation then yields
$$
\E[M_t- M_s\,|\,\mathcal{F}_s] \leq 0.
$$
This establishes inequality (b) and the proof is complete.
\end{proof}
Finally, we are ready to prove the verification theorem.
\subsubsection*{Proof of Verification Theorem \ref{verthm}.}
Since $V_\pi$ is twice continuously differentiable
and dominated by an affine function, the function
$h(x,l):=V_\pi(x,l)$  is in the domain of the extended generator of~$X=(R,Y)$. This means that the
process
$$V_\pi(R_{t\wedge T_I},Y_{t\wedge T_I})e^{-\int_0^{t\wedge T_I} \frac{\mathfrak{A} h(X_s)}{h(X_s)}\md s}
$$
is a martingale for any closed interval $I\in[0,\infty )$. By \eqref{HJB} it follows that
$\frac{\mathfrak{A}h(X_s)}{h(X_s)}\leq 0$ and hence $V_\pi\in\mathcal{G}$, which completes the proof.
\hfill\mbox{$\square$}
\subsection{On the return function for a barrier strategy}
\subsubsection*{Proof of Theorem \ref{barthm}.}
Note that for the barrier strategy until the first hitting of the barrier $a$, the regulated process $R^{\pi_a}$ behaves like the process $R$. By the strong Markov property of $R_t$ and by (\ref{exit1}) for $x\in[0,a]$ we have
$$
V_a(x,l)=
\frac{\Delta (x)}{\Delta (a)} V_a(a,l) + \E_{(x,l)}\Big[e^{-Y_{\tau^-_0}}w(R_{\tau_0^-})\one_{[ \tau_0^-<\tau^+_a]}\Big].
$$
Moreover, again using the strong Markov property and \eqref{exit1} we can derive
$$\E_{(x,l)}\Big[e^{-Y_{\tau^-_0}}w(R_{\tau_0^-})\one_{[ \tau_0^-<\tau^+_a]}\Big]=\Big(\gqw (x)-\gqw(a)\frac{\Delta (x)}{\Delta (a)}\Big)e^{-l}\;.$$
Hence
\begin{equation*}
V_a(x,l)=
\frac{\Delta (x)}{\Delta (a)}\big( V_a(a,l)-e^{-l}\gqw (a)\big)+e^{-l}\gqw (x).
\end{equation*}
Note that $L_t^{\pi_a}=(\sup_{s\leq t} R_s -a)\vee 0$. Thus using the classical arguments for the L\'evy dividend problem,
(see e.g. Avram et al. \cite[eq. (5.12)]{APP2}) it follows that
\begin{equation*}
\frac{\md}{\md x}V_a(a,l)=e^{-l},
\end{equation*}
from which the assertion of Theorem \ref{barthm} immediately follows.
\hfill\mbox{$\square$}
\subsection{Proofs of necessary and sufficient conditions for optimality of a barrier strategy\label{proof:2}}
\subsubsection*{Proof of Theorem \ref{ver2}.}
\noindent
To prove sufficiency, we need to show that $V_{a^*}$ satisfies the conditions of the Verification Theorem \ref{verthm}.
From Theorem \ref{barthm} it follows that $V_{a^*}$ is ultimately linear and by Assumption \eqref{condW1} 
is twice continuously differentiable.
Moreover, by the choice of the optimal barrier $a^*$ we know that $V_{a^*}^\prime (x)\geq 1$.
Finally,
by definition of  $\Delta$ and $G_{w}$ in \eqref{exit1} and \eqref{Gqw} respectively, and the strong Markov property of the risk process $R$ it follows that
$$
e^{-Y_{t\wedge T}}\Delta(R_{t\wedge T\wedge \tau^+_{a^*}}),\qquad e^{-Y_{t\wedge T}}G_{w}(R_{t\wedge T}) $$ are martingales. Hence
$$e^{-Y_{t\wedge T}}F_{a^*}(R_{t\wedge T \wedge \tau^+_{a^*}})$$ is a martingale.
This means that $F_{a^*}$ is in the domain of the full generator of $R$ stopped on exiting $[0,a^*]$ and
that $\mathfrak{A} (F_{a^*}(x)e^{-l})= 0$ for $x\leq a^*$ and $l\in\mathbb{R}$.
The remaining part of HJB equation follows from assumption \eqref{inq2}.
\smallskip
\\
To prove necessity we assume that condition (\ref{inq2}) is not satisfied. By the continuity of the function $x\mapsto \mathfrak{A}(F_{a^*}(x)e^{-l})$ there exists an open and bounded interval $\mathrm{J}\subset(a^*,\infty)$ such that $\mathfrak{A} (F_{a^*}(x)e^{-l})> 0$ for all $x\in \mathrm{J}$. Let $\tilde{\pi}$ be the strategy of paying nothing if the reserve process $R^{\tilde{\pi}}$ takes a value in $\mathrm{J}$, and following the strategy $\pi_{a^*}$ otherwise. If we extend $V_{a^*}$ to the negative half-axis by $F_{a^*}(x)=w(x)$ for $x<0$, we have
\begin{equation*}
 V_{\tilde{\pi}}(x,l) =
\begin{cases}
\E_{(x,l)}[e^{-Y_{T_\mathrm{J}}}F_{a^*}(R_{T_\mathrm{J}})], & x\in \mathrm{J}, \\
e^{-l}F_{a^*}(x), & x\not\in \mathrm{J},
\end{cases}
\end{equation*}
where $T_\mathrm{J}$ is defined by (\ref{eq:ggmartI}).
\smallskip
\\By the Optional Stopping Theorem applied to the process $e^{-Y_t}F_{a^*}(R_t)$, for all $x\in \mathrm{J}$, we obtain
\begin{equation*}
V_{\tilde{\pi}}(x,l)=\E_{(x,l)}[e^{-Y_{T_\mathrm{J}}}F_{a^*}(R_{ T_\mathrm{J}})]=e^{-l}F_{a^*}(x)+\E_{(x,l)}\bigg[\int_0^{ T_\mathrm{J}}\mathfrak{A} (F_{a^*}(R_s)e^{-Y_s})\md s\bigg]> e^{-l}F_{a^*}(x).
\end{equation*}
This leads to a contradiction and consequently proves the optimality of the strategy $\pi_{a^*}$.
\hfill\mbox{$\square$}
\subsubsection*{Proof of Theorem \ref{suffcond}.\label{proof:3}}
In the first step, we will show that
\begin{equation}\label{pom1}
\lim_{y\uparrow x} \,\mathfrak{A}[(F_{a^*}-F_x)(y)e^{-l}]\leq 0\qquad\text{for all }x>a^*\;, l\in\R\;.
\end{equation}
Let $x>a^*$. By the Dominated Convergence Theorem we obtain
\begin{align*}
 \lim_{y\uparrow x}\,\mathfrak{A}[(F_{a^*}-F_x)(y)e^{-l}]
&=e^{-l}\Big\{c-\rho\delta\sigma-\int_{\R^2}h_2\one_{[|h_2|\le 1]}\;\nu(\td h_1,\td h_2)\Big\}\big(F'_{a^*}-F'_x\big)(x)
\\&{}-e^{-l}\Big\{p-\frac{\delta^2}2+\int_{\R^2}1-h_1\one_{[|h_1|\le 1]}\;\nu(\td h_1,\td h_2)\Big\}\big(F_{a^*}-F_x)(x)\big)
\\
&{} +e^{-l}\int_{\R^2} e^{-h_1}\big[(F_{a^*}-F_x)(x+h_2)\big]\;\nu(\td h_1,\td h_2).
\\&+e^{-l}\frac{\sigma^2}{2}\big(F^{''}_{a^*}(x)-\lim_{y\uparrow x}F^{''}_x(y)\big)\;.
\end{align*}
By \eqref{va} we have for $x>a^*$:
\begin{itemize}
 \item[i.] $(F'_{a^*}-F'_x)(x)=0$.
\item[ii.]  $(F'_{a^*}-F'_x)(b)=\Delta (b) \left(\hp(a^*)-\hp(x)\right)\geq0$ for $b\in [0,a^*]$ by the definition of $a^*$.
\item[iii.]  $(F'_{a^*}-F'_x)(u)=\Delta(u)\left(\hp(u)-\hp(x)\right)\geq0$  for $u\in [a^*,x]$ by the assumption \eqref{as1}.
\item[iv.] $(F_{a^*}-F_x)(a^*)\geq0$, thus by iii, $(F_{a^*}-F_x)(x)\geq0$.
\item[v.] $(F_{a^*}-F_x)(x+z)\leq(F_{a^*}-F_x)(x)$ for all $z\leq0$ by ii and iii.
\item[vi.] Assumption \eqref{assump:2} yields $-p+\frac{\delta^2}2+\int_{\R^2}-1+h_1\one_{[|h_1|\le 1]}\;\nu(\td h_1,\td h_2)<-\int_{\R^2}e^{-h_1}\;\nu(\td h_1,\td h_2)$.
\item[vii.] If $\sigma>0$ then by our assumption \eqref{as1} we have $\lim_{y\uparrow x}F^{''}_x(y)\geq 0=F^{''}_{a^*}(x)$.
\end{itemize}
Thus we have shown \eqref{pom1}.
\\
Now assume that \eqref{inq2} does not hold. Then there exists an $x>a^*$ such that $$\mathfrak{A} (F_{a^*}(x)e^{-l})>0\;. $$ By the continuity of $\mathfrak{A} (F_{a^*}e^{-l})$ we deduce that $\lim_{y\uparrow x}\mathfrak{A}(F_x(y)e^{-l})>0$, which contradicts \eqref{pom1}.
\hfill\mbox{$\square$}
\subsubsection*{Proof of Corollary \ref{cor:explopt}.\label{proof:4}}
\noindent
It is well known that the scale function of a spectrally negative L\'evy
process which does not go to minus infinity is equal (up to a multiplicative
constant appearing in the local time) to the renewal function of the
descending ladder height process. Following \cite{Loeffen1} and Assumption \eqref{eq:condW} we
conclude that $W^\alpha (x)$ is completely monotone, confer the footnote on p.\ \ref{foot} for definition, and, as it is non-negative, it is also a Bernstein function.
Thus (see, e.g., \cite{20}, Chapter 3.9):
\[W^\alpha(x) = a+ bx +
\int_0^\infty (1-e^{-xt})\;\xi(\td t),\quad x > 0,\]
where $a, b \geq 0$ and $\xi$ is a measure on $(0,\infty)$ satisfying
integrability condition:
$$\int_0^\infty (t \wedge 1)\;\xi(\td t) <\infty\;.$$
From Lemma \ref{Wrepr} it follows that
\begin{equation*}
\Delta(x) =\left( e^{\alpha x}(a + bx)+
\int_0^\infty
\left(e^{\alpha x} -e^{-x(t-\alpha)}\right)\;\xi(\td t)\right).
\end{equation*}
By repeatedly using the dominated convergence theorem, we can now deduce
\begin{eqnarray*}
\Delta'''(x) =
\left[g'''(x)+
\int_0^\infty
\left(\alpha^3 e^{\alpha x}x +(t-\alpha)^3 e^{-x(t-\alpha)}\right)\xi(\td t)\right],
\end{eqnarray*}
where $g(x) = e^{\alpha x}(a+bx)$. Hence $\Delta'''(x)>0$ for all $x>0$ and so $\Delta'(x)$ is strictly convex on $(0,\infty)$. We can
now apply Theorem \ref{suffcond} to deduce that the barrier strategy at $a^*$ is optimal.
\section*{Conclusions}\label{sec:remarks}
In this paper, we sought to maximise the amount of expected dividends paid in a foreign currency for dependent L\'evy risk processes as a surplus and exchange rate. We found some sufficient and necessary conditions for a constant barrier strategy to be optimal.
\\It would be interesting to consider a more general exchange rates $Y_t$, for examples ones governed by stochastic differential equations. In Section Examples, we demonstrated how the optimal strategy and the value function can be found through direct solution of the HJB equation (classical method) and via scale functions (the method presented in this paper). Of course, solution via classical methods by guessing a twice continuously differentiable function solving the HJB equation, can be applied in just few cases by dealing with L\'evy processes, for instance if the jumps are assumed to be exponentially distributed. In the remaining cases one has to rely on the method presented in this paper.
\subsubsection*{Acknowledgements}
Julia Eisenberg has been supported by the Austrian Science Fund, FWF, V 603-N35.\smallskip
\\Zbigniew Palmowski has been partially supported by the National Science Centre under the grant 2016/23/B/HS4/00566.

\end{document}